\documentclass[a4paper,UKenglish]{llncs}

\usepackage{microtype}
\usepackage[utf8]{inputenc}
\usepackage{amssymb,amsmath,stmaryrd,appendix,amstext}
\usepackage{caption}
\usepackage{subfig}
\usepackage{enumerate, paralist, enumitem}
\usepackage{tikz} 

\spnewtheorem{customlem}{Lemma}{\bfseries}{\rmfamily}

\spnewtheorem{redrule}{Reduction}{\bfseries}{\rmfamily}
\spnewtheorem{algoss}{Algorithm}{\bfseries}{\rmfamily}
\spnewtheorem{observation}{Observation}{\bfseries}{\rmfamily}

\DeclareMathOperator{\operatorClassNP}{NP}
\newcommand{\classNP}{\ensuremath{\operatorClassNP}}
\DeclareMathOperator{\operatorClassCoNP}{coNP}
\newcommand{\classCoNP}{\ensuremath{\operatorClassCoNP}}

\newcommand{\G}{\mathcal{G}}

\newcommand{\Q}{\mathcal{Q}}
\newcommand{\C}{\mathcal{C}}

\newcommand{\N}{\mathbb{N}}

\newcommand{\bigO}{\mathcal{O}}
\newcommand{\cO}{\mathcal{O}}

\renewcommand{\phi}{\varphi}
 
\newcommand{\yes}{\textit{\text{yes}}}
\newcommand{\VCMF}{\textsc{Vertex Cover/Mock forest modulator}}
\newcommand{\VCMFshort}{\textsc{VC/MFM}}
\newcommand{\VCPF}{\textsc{Vertex Cover/Pseudoforest modulator}}
\newcommand{\VCPFshort}{\textsc{VC/PFM}}
\newcommand{\ISMF}{\textsc{Independent Set/Mock forest modulator}}
\newcommand{\ISMFshort}{\textsc{IS/MFM}}
\newcommand{\ISPF}{\textsc{Independent Set/Pseudoforest modulator}}
\newcommand{\ISPFshort}{\textsc{IS/PFM}}
\newcommand{\ISFVS}{\textsc{Independent Set/Feedback Vertex Set}}
\newcommand{\ISFVSshort}{\textsc{IS/FVS}}
\newcommand{\VCproblem}{\textsc{Vertex Cover}}
\newcommand{\ISproblem}{\textsc{Independent Set}}
\newcommand{\FVSproblem}{\textsc{Feedback Vertex Set}}
\newcommand{\CNF}{\textsc{CNF-SAT}}
\newcommand{\true}{\textsc{True}}
\newcommand{\false}{\textsc{False}}
\newcommand{\NPincoNPpoly}{$\text{NP}\subseteq \text{coNP/poly}$}
\newcommand{\CONF}{\textsc{Conf}}

\usepackage{xspace}
\newcommand{\ie}{i.\,e.\@\xspace}
\newcommand{\eg}{e.\,g.\@\xspace}

\newcommand{\no}{\textit{\textbf{no}}}

\newcommand{\defparproblem}[4]{
  \vspace{1mm}
\noindent\fbox{
  \begin{minipage}{0.96\textwidth}
    \begin{tabular*}{\textwidth}{@{\extracolsep{\fill}}lr} #1 
    \\ 
      \end{tabular*}
    {\textbf{Input:}} #2  \\
    {\textbf{Parameter:}} #3  \\
    {\textbf{Question:}} #4
  \end{minipage}
  }
  \vspace{1mm}
}

\begin{document}

\title{Vertex Cover Structural Parameterization Revisited\thanks{Supported by Rigorous Theory of Preprocessing, ERC
    Advanced Investigator Grant 267959.}}
	
\author{Fedor V. Fomin \and Torstein J. F. Str{\o}mme}

\institute{
  Department of Informatics, University of Bergen, Norway,\\
  \texttt{\{fedor.fomin,torstein.stromme\}@ii.uib.no}
}

\pagestyle{plain}
\maketitle

\title{Vertex Cover Structural Parameterization Revisited\footnote{Supported by Rigorous Theory of Preprocessing, ERC
    Advanced Investigator Grant 267959.}}

\authorrunning{F.\ V.\ Fomin and T.\ J.\ F.\ Str{\o}mme}

%
\begin{abstract}
A pseudoforest is a graph whose connected components have at most one cycle. Let $X$ be a pseudoforest modulator of graph $G$, \ie{} a vertex subset of $G$ such that $G-X$ is a pseudoforest. We show that \VCproblem{} admits a polynomial kernel being parameterized by the size of the pseudoforest modulator. In other words, we provide a polynomial time algorithm that for an input graph $G$ and integer $k$, outputs a graph $G'$ and integer $k'$, such that $G'$ has $\cO(|X|^{12})$ vertices and $G$ has a vertex cover of size $k$ if and only if $G'$ has vertex cover of size $k'$.
We complement our findings by proving that 
 there is no polynomial kernel for \VCproblem{} parameterized by the size of a modulator to a \emph{mock forest} (a graph where no cycles share a vertex) unless \NPincoNPpoly{}. In particular, this  also rules out polynomial kernels when parameterized by the size of a modulator to outerplanar   and cactus graphs.

\end{abstract}
%
\section{Introduction}\label{sec:intro}
Kernelization is  a fundamental algorithmic methodology rooted in parameterized complexity. 
It also serves as a  
rigorous mathematical tool for analyzing   certain polynomial-time preprocessing or data-reductions algorithms.
  In this paper we provide new kernelization algorithm for ``structural'' parameterization of  \VCproblem.

In the \VCproblem{} problem, we are given as input a graph $G$ and a positive
integer $k$, and are asked if there exists a set $S$ of at most
$k$ vertices in $G$ such that every edge in $G$ is adjacent to
at least one of the vertices in $S$; such an $S$ is called a \emph{vertex
cover} of $G$. 
 As a part of a general program on kernelization with structural parameterization, 
Jansen and Bodlaender   
\cite{JansenB11} initiated  the study of kernelization for \VCproblem{} with ``refined'' parameterization 
 by showing that  it admits a polynomial kernel when parameterized by the size of a feedback vertex set, i.e. a forest-modulator. Since a feedback vertex set can be significantly smaller than a vertex cover,  in various  situations such a kernel can  be preferable.
  
  It is a very natural question if the  kernelization  result  of Jansen and Bodlaender can be extended to parameters which are ``stronger'' than the size of a feedback vertex set. 
  Forests are exactly the graphs of treewidth one and a natural direction of such an extension would be to explore the parameterization by a  constant treewidth modulator. However, 
%
  as  it was shown by Cygan et al. \cite{CyganLPPS14},  for each $t\geq 2$,  \VCproblem{} does not admit a polynomial kernel being parameterized by the size of the treewidth $t$ modulator unless 
 $\classNP\subseteq\classCoNP/\text{\rm poly}$. Since
the result of Cygan et al. \cite{CyganLPPS14} rules out polynomial kernels for \VCproblem{} even when  parameterized by treewidth-2 modulators,
the next natural step in the study of polynomial kernelization for \VCproblem{}  is to see if the problem admits a polynomial kernel when parameterized by a modulator to some subclasses of graphs of treewidth $2$. Towards this end, Majumdar, Raman and Saurabh \cite{majumdar_et_al:LIPIcs:2015:5594} obtain a polynomial kernel for \VCproblem{} parameterized by the size of a degree-2 modulator. 

In this work we show that  \VCproblem{} admits a polynomial kernel when the parameter is the size of a pseudoforest modulator. More precisely, a \emph{pseudoforest} is an undirected graph in which every connected component has at most one cycle. In a graph $G$, a vertex set $X$ is a \emph{pseudoforest modulator} if the graph $G-X$ obtained from $G$ by deleting $X$ is a pseudoforest.
 We define the following problem

\defparproblem{\VCPF{} (\VCPFshort)}{A simple undirected graph $G$, a pseudoforest-modulator set $X\subseteq V(G)$ such that $G-X$ is a pseudoforest, integer $k$.}{Size of a pseudoforest modulator $|X|$.}{Does $G$ contain a vertex cover of size at most $k$?}

\medskip\noindent\textbf{Our results.}
We  show that \VCPFshort {} admits a polynomial kernel with $\cO(|X|^{12})$ vertices. Since every feedback vertex set is  a pseudoforest-modulator  and every degree-2-modulator is also a pseudoforest-modulator, our result extends the borders of polynomial  kernelization  for \VCproblem{}  established by  Jansen and Bodlaender   
\cite{JansenB11} and by Majumdar et. al. \cite{majumdar_et_al:LIPIcs:2015:5594}.

We complement our kernelization algorithm with the lower bound.  Let us observe that the work of Cygan et al. \cite{CyganLPPS14} does not rule out the existence of polynomial kernels when the problem is parameterized by the size of a modulator to some proper subclass of treewidth-2 graphs, like outerplanar graphs or cactus graphs, \ie{} graphs where every 2-connected component is a cycle.  
We refine the known lower bounds by proving that a  polynomial kernel for  \VCproblem{}  parameterized by  the size of mock forest modulator would imply  $\classNP\subseteq\classCoNP/\text{\rm poly}$.  (\emph{Mock forest}  is a graph with no two cycles  sharing a vertex and thus of treewidth at most $2$.) Since a mock forest is also   outerplanar and   cactus graph, this rules out polynomial kernels 
parameterized by the size of a modulator to these classes as well.

While we state our kernelization  result assuming that a pseudoforest modulator is given as a part of the input, this condition can be omitted. There are several   approximation algorithms for 
 pseudoforest modulator. For example, computing a modulator to a  pseudoforest is a special case of the \textsc{$\mathcal{F}$-Deletion} problem considered in \cite{FominLMS12}, and there is a 
 randomized constant factor approximation algorithm of running time $\cO (nm)$. Also since pseudoforests of a graph form independent sets of a 
 bicircular matroid, it follows  from the generic framework of Fujito
\cite{Fujito98}  that there is a 
deterministic polynomial time $2$-approximation algorithm for pseudoforest modulator. 


The proof of our main result is constructive and consists of several reduction rules. While some of our steps follow Jansen and Bodlaender   
\cite{JansenB11}, the  essential part of the proof is different. Our algorithm is based on a novel combinatorial result about  maximum independent sets  in  pseudotrees  (Lemma~\ref{lem:pseudotree-small-Z}), which is also interesting in its own. 

The remaining part of the paper is organized as follows. In Section~\ref{sec:prel} we give preliminaries on notation and definitions. In Section~\ref{sec:kern} we develop the kernelization algorithm for \VCPF, which is the main content of this paper. The section containing the proof of Lemma~\ref{lem:pseudotree-small-Z} is quite technical and is found in Section~\ref{sec:pseudotree-small-Z}. We obtain lower bounds for \VCproblem{}   parameterized by the vertex deletion distances to mock forests.  

%

%
\section{Preliminaries}\label{sec:prel}

\noindent\textbf{Graph theoretic notions.} 

In this paper we are concerned only with finite, simple, loopless, undirected graphs. A graph $G$ consists of a set of vertices $V(G)$ and a set of edges $E(G)$. For a vertex $v \in V(G)$, its neighborhood $N(v)$ is the set of all vertices adjacent to $v$. The closed neighborhood of $v$ is denoted $N[v] = N(v) \cup \{v\}$. Similarly, for a set $S \subseteq V(G)$ we denote its neighborhood $N(S) = (\bigcup_{v\in S} N(v)) \setminus S$, and its closed neighborhood $N[S] = N(S) \cup S$. The degree of a vertex is the number of vertices adjacent to $v$ in the graph, $deg(v) = |N(v)|$. A vertex of degree $1$ is called a leaf. In cases where it may be unclear which graph is being referred to, a subscript is added, e.g.\@ $N_{G'}(S)$ denotes the neighborhood of $S$ in the graph $G'$.

A subgraph $G' \subseteq G$ is a graph such that $V(G') \subseteq V(G)$ and $E(G') \subseteq E(G)$. For a set $S \subseteq V(G)$, the subgraph induced by the vertices of $S$ is denoted $G[S]$. The graph where $S$ and its incident edges are removed, is denoted as $G-S = G[V(G)\setminus S]$. Similarly, the graph $G[V(G) \setminus \{v\}]$ obtained by removing a single vertex $v$ and its incident edges is denoted $G-v$, and the graph obtained by removing a subgraph $G' \subseteq G$ and its incident edges is denoted $G-G' = G[V(G)\setminus V(G')]$.

A tree $T$ is a connected graph which contains no cycles. A tree is \emph{rooted} if one vertex $r \in V(T)$ has been designated as the root. In rooted trees, all vertices have a natural orientation with respect to the root. For a non-leaf vertex $a \in V(T)$ we denote the set of its children by $C(a)$. For two vertices $a,b \in V(T)$, we say that $a$ is an ancestor of $b$ if $a$ is on the path from $r$ to $b$ (by this definition, a vertex is always an ancestor of itself). We let the subtree rooted at the vertex $a$ be denoted by $T_a = T[\{b \text{ }| \text{ } a \text{ is an ancestor of }b\}]$. A subtree $T_a$ is a \emph{strict} subtree if $T \neq T_a$, in other words if $a \neq r$. A graph $F$ is a \emph{forest} if every connected component of $F$ is a tree.

An \emph{independent set} of $G$ is a set $I \subseteq V(G)$ such that every edge of $G$ has at most one endpoint in $I$. We let $\alpha(G)$ denote the independence number of $G$, \ie{} the largest number of distinct vertices which can constitute an independent set of $G$. An independent set $I$ of $G$ with size $|I| = \alpha(G)$ is called a maximum independent set, abbreviated MIS.

A \emph{feedback vertex set} (FVS) of $G$ is a set $X \subseteq V(G)$ such that $G-X$ is a forest. A graph $G$ is a \emph{mock forest} if no cycles of $G$ share a vertex.

\medskip
\noindent\textbf{Kernels and reductions.} 
 \begin{definition}[Kernelization, kernel]
A {\em{kernelization algorithm}}, or simply a {\em{kernel}}, for a parameterized problem $Q$ is an algorithm $\mathcal{A}$ that, given an instance $(I,k)$ of $Q$, works in polynomial time and returns an equivalent instance $(I',k')$ of $Q$. Moreover, we require that $k' + |I'| \leq g(k)$ for some computable function $g\colon \mathbb{N}\to \mathbb{N}$.
\end{definition}
If the upper bound $g(\cdot)$ is a polynomial  function of the parameter, then we say that $Q$ admits a {\emph{polynomial   kernel}}.


%
\section{Kernelization}\label{sec:kern}

This section contains the kernelization of \VCPF\ and is the main section of the paper. We will first develop a kernel for \ISPF, and then by the immediate correspondence between the \VCproblem\ and \ISproblem\ problems the kernel for \VCPFshort\ will follow. For the remainder of this section, we will thus focus on the \ISPF\ problem:

\defparproblem{\ISPF{} (\ISPFshort)}{A simple undirected graph $G$, a pseudoforest-modulator set $X\subseteq V(G)$ such that $G-X$ is a pseudoforest, integer $k$.}{Size of a pseudoforest modulator $|X|$.}{Does $G$ contain an independent set of size at least $k$?}

Throughout the section, let $F := G - X$ be the induced subgraph remaining after the modulator $X$ has been removed from $G$. Note that $F$ is a pseudoforest.

We say that $(G,X,k)$ is a \emph{\yes-instance} of \ISPFshort\ if there exists an independent set $I$ of $G$ such that $|I| \geq k$. We say it is a \emph{\no-instance} if there is no such set.

\begin{definition} \label{df:conflicts}
    (Conflicts) Let $(G,X,k)$ be an instance of \ISPFshort\ where $F' \subseteq F$ is a subgraph of the pseudoforest $F$ and $X' \subseteq X$ is a subset of the modulator $X$. Then the number of \emph{conflicts} induced by $X'$ on $F'$ is defined as $\CONF_{F'}(X') := \alpha(F') - \alpha(F' - N_	G(X'))$.
\end{definition}

\begin{figure}[ht] \label{pfm:fig:conflicts}
    \centering
    \includegraphics[scale=1.0]{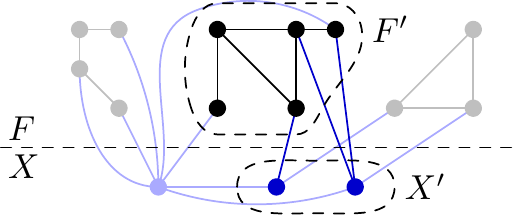}
    \caption{Conflicts: In the figure, we observe that $\alpha(F') = 3$, and $\alpha(F'-N_G(X')) = 1$. Hence, we get that $\CONF_{F'}(X') = 2$. In other words, the number of conflicts induced by $X'$ on $F'$ is $2$.}
\end{figure}

Choosing $X'$ to be in the independent set $I$ of $G$ may prevent some vertices in $F'$ from being included in same set $I$. In particular, no vertex $v \in V(F') \cap N_G(X')$ can be chosen to be in $I$. In light of this, the term $\CONF_{F'}(X')$ can be understood as the price one has to pay in $F'$ by choosing to include $X'$ in the independent set.

Observe that $\CONF_{F'}(X')$ is polynomial time computable when the independence number $\alpha(F')$ is.

\begin{definition} \label{df:chunks}
    (Chunks) Let $(G,X,k)$ be an instance of \ISPFshort. A set $X' \subseteq X$ is a \emph{chunk} if the following hold:
    \begin{itemize}
        \item $X'$ is independent in $G$,
        \item The size of $X$ is between 1 and 3, \ie{} $1 \leq |X| \leq 3$, and
        \item The number of conflicts induced by $X'$ on the pseudoforest $F$ is less than $|X|$, \ie{} $\CONF_F(X') < |X|$.
    \end{itemize}
    We let $\mathcal{X}$ be the collection of all chunks of $X$.
\end{definition}

The collection of chunks $\mathcal{X}$ can be seen as all suitable candidate subsets of size at most 3 from $X$ to be included in a maximum independent set $I$ for $G$. The idea is that $I$ may contain a chunk as a subset, but need not include a subset $X' \subseteq X$ of size at most 3 which is \emph{not} a chunk. This will allow us to discard potential solutions containing non-chunk subsets of $X$ with size at most 3. In order for this intuition to hold, we provide the following lemma, originally by Jansen and Bodlaender \cite[Lemma 2]{JansenB11} though slightly altered to fit our purposes.

\begin{lemma} \label{lem:forbidden}
    If there exists an independent set of size $k$ in $G$, then there exists an independent set $I$ of $G$ such that $|I| \geq k$ and for all subsets $X' \subseteq X \cap I$, $\CONF_F(X') < |X|$.
\end{lemma}

\begin{proof}
    Assume that $I' \subseteq V(G)$ is an independent set for which there exists some $X' \subseteq I' \cap X$ such that $CONF_F(X') \geq |X|$. We will show that then there is also another independent set $I$ such that $|I| \geq |I'|$ and for all subsets $X'' \subseteq X$ with $CONF_F(X'') \geq |X|$, $X''$ is not a subset of $I$.
    
    Because $CONF_F(X') \geq |X|$, we have that $\alpha(F) \geq |X| + \alpha(F-N_G(X'))$. Now consider the independent set $I'$. Some of its vertices are in $X$, however no more than $|X|$. The remainder of vertices of $I'$ are in $V(F)$, however no more than $\alpha(F-N_G(X'))$. Thus $|X| + \alpha(F-N_G(X')) \geq |I'|$. But then $\alpha(F) \geq |I'|$, and we see that a MIS of the pseudoforest $F$ satisfies the requirements of the lemma.
\end{proof}

\begin{definition} \label{def:anchortriangle}
    (Anchor triangle) Let $(G, X, k)$ be an instance of \ISPFshort. Let $P$ be a connected component in $F$ with $V(P) = \{p_1, p_2, p_3\}$. Then $P$ is an \emph{anchor triangle} if there exists a set $\{x_1,x_2,x_3\} \subseteq X$ such that:
    \begin{itemize}
        \item $N_G(p_1) = \{p_2, p_3, x_1\}$
        \item $N_G(p_2) = \{p_1, p_3, x_2\}$
        \item $N_G(p_3) = \{p_1, p_2, x_3\}$
    \end{itemize}
    An anchor triangle is \emph{non-redundant} if there is no other anchor triangle with the same open neighborhood in $G$. 
\end{definition}

\begin{figure}[ht] \label{fig:anchortriangle}
	\centering
    \includegraphics{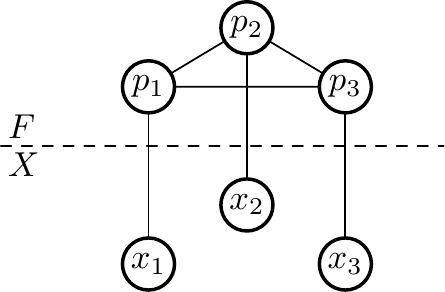}
    \caption{The connected component $P \subseteq F$ with vertices $V(P) = \{p_1, p_2, p_3\}$ is an anchor triangle for the triple $\{x_1, x_2, x_3\} \subseteq X$.}
    \vspace{-.7cm}
\end{figure}

\begin{definition} \label{def:unnecessary}
    (Unnecessary triple) A triple ${_3X} \subseteq X$ is said to be \emph{unnecessary} if there exists an anchor triangle $P$ such that $N_G(P) = {_3X}$.
\end{definition}

The fact that a triple ${_3X} \subseteq X$ is unnecessary as defined above should intuitively be understood with respect to constructing an independent set. If a triple ${_3X}$ is unnecessary then there exists a MIS which does not contain all of ${_3X}$. This intuition is supported by the next lemma.

\begin{lemma} \label{lem:unnecessary}
    Let $(G, X, k)$ be an instance of \ISPFshort. If there exists an independent set of size at least $k$ in $G$, then there exists an independent set $I$ of $G$ with $|I| \geq k$ containing no unnecessary triple ${_3X} \subseteq X$.
\end{lemma}

\begin{proof}
    We will prove the lemma by constructing the required independent set $I$ satisfying the properties, assuming we have an arbitrary independent set $I'$ of size at least $k$. By definition, for each unnecessary triple ${_3X} \subseteq X \cap I'$, there is at least one anchor triangle $P \subseteq F$ such that $N_G(P) = {_3X}$. Since ${_3X} \subseteq I'$, we have that no vertex of $P$ is in $I'$. Let $p_u \in V(P)$ be an arbitrary vertex of $P$, and let $u \in {_3X}$ be the neighbor of $p_u$ in ${_3X}$. Observe that we can here replace $u$ by $p_u$ in $I'$, and still have $I'$ be an independent set of the same size. We do this for every unnecessary triple ${_3X} \subseteq X \cap I'$ to obtain $I$, which then satisfies the requirement of the lemma.
\end{proof}

\subsection{Reduction Rules} \label{ssec:rr}

We introduce here the reduction rules. Each reduction receives as input an instance $(G, X, k)$ of \ISPF, and outputs an equivalent instance $(G', X', k')$. A reduction is \emph{safe} if the input and output instances are equivalent, that is, $(G, X, k)$ is a \yes-instance if and only if $(G', X', k')$ is a \yes-instance. Reductions \ref{rr:singlesRemoved}, \ref{rr:doublesForbidden} and \ref{rr:ccRemoved} originates in \cite{JansenB11}, though Reduction \ref{rr:ccRemoved} is altered to fit the context of a pseudoforest, which also required some changes to the proof.

Reduction rules will be applied exhaustively starting with lower number rules, until Reduction \ref{rr:ccRemoved} is no longer applicable. During this process, a lower number rule is always applied before a higher number rule if at any point they are both applicable. Then Reductions \ref{rr:split-cycles} and \ref{rr:reducedinstance-from-IS-by-FVS} will be applied once each to obtain the final reduced instance. Note that each reduction is computable in polynomial time.

\begin{redrule} \label{rr:singlesRemoved}
    If there is a vertex $v \in X$ such that $\CONF_F(\{v\}) \geq |X|$, then delete $v$ from the graph $G$ and from the set $X$. We let $G' := G-v$, $X' := X-v$ and $k' := k$.
\end{redrule}

\begin{redrule} \label{rr:doublesForbidden}
    If there are distinct vertices $u,v\in X$ with $uv \notin E(G)$ for which $\CONF_F(\{u,v\}) \geq |X|$, then add edge $uv$ to $G$. We let $G' := (V(G), E(G) \cup \{uv\})$, $X' := X$ and $k' := k$.
\end{redrule}

Reductions~\ref{rr:singlesRemoved} and~\ref{rr:doublesForbidden} are safe due to Lemma~\ref{lem:forbidden}.

\begin{redrule} \label{rr:triplesOptional}
    If there are distinct $u,v,w \in X$ such that $\CONF_F(\{u,v,w\}) \geq |X|$, the set $\{u,v,w\}$ is independent in $G$, and for which there is no anchor triangle $P$ with $N(P) = \{u,v,w\}$, then add an anchor triangle $P' = \{p_u,p_v,p_w\}$ to the graph such that $N(P') = \{u,v,w\}$, and increase $k$ by one. Let $V(G') := V(G) \cup \{p_u,p_v,p_w\}$ and let $E(G') := E(G) \cup \{p_up_v, p_up_w, p_vp_w, p_uu, p_vv, p_ww\})$. Further, let $X' := X$ and let $k' := k+1$.
\end{redrule}
\vspace{-.3cm}
\begin{figure}[ht] \label{pfm:fig:rrtriplesOptional}
    \centering
    \includegraphics[width=.48\linewidth]{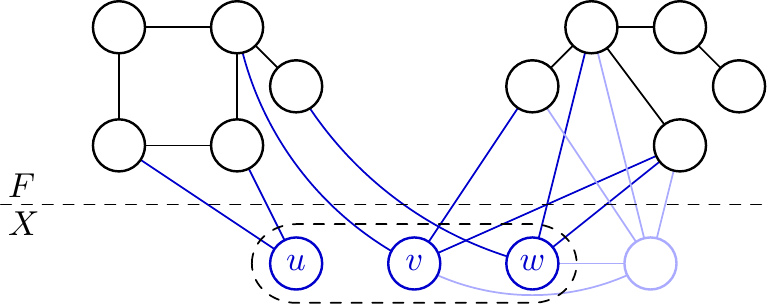}
	\hspace*{.02\linewidth}
	\includegraphics[width=.48\linewidth]{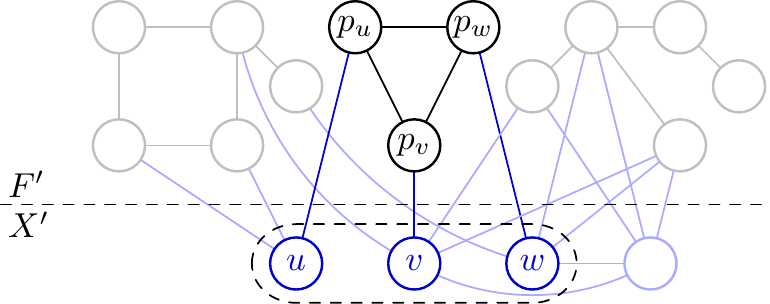}
    \caption{Reduction~\ref{rr:triplesOptional}: Adding an anchor triangle to the independent triple $\{u,v,w\}$ ($k' = k+1$). This makes $\{u,v,w\}$ an unnecessary triple in the output instance.}
\end{figure}
\vspace{-.4cm}
Note that Reduction \ref{rr:triplesOptional} makes the triple $\{u,v,w\} \subseteq X$ unnecessary in the reduced instance as defined in Definition \ref{def:unnecessary}.

\begin{lemma} \label{lem:triplesOptional}
    Reduction \ref{rr:triplesOptional} is safe. Let $(G,X,k)$ be an instance of \ISPFshort\ to which Reduction \ref{rr:triplesOptional} is applicable, and let $(G', X', k')$ be the reduced instance. Then $(G,X,k)$ is a \yes-instance if and only if $(G', X', k')$ is a \yes-instance.
\end{lemma}

\begin{proof}
    For the forward direction of the proof, assume $(G,X,k)$ is a \yes-instance, let $u,v,w \in X$ be the elements which triggers the reduction, and let $P' = \{p_u,p_v,p_w\}$ be the vertices of the added anchor triangle in the reduced instance. Let $I$ be an independent set of $G$ with $|I| \geq k$. By Lemma~\ref{lem:forbidden}, we assume that at least one of $u,v,w$ is not in $I$. Without loss of generality, let $u \notin I$. For the reduced instance, observe that $I' := I \cup \{p_u\}$ is an independent set of $G'$ with size $|I'| = |I| + 1 \geq k+1 = k'$, which makes $(G',X',k')$ a \yes-instance.
    
    For the backward direction, assume that $(G',X',k')$ is a \yes-instance, and let $P'$ be the added anchor triangle in the reduced instance. Let $I'$ be an independent set of $G'$ with $|I'| \geq k'$. Because $P'$ induces a triangle in $G'$, at most one vertex of $P'$ is in $I'$. Thus, $I := I' \setminus P'$ is an independent set of $G$ with $|I| \geq |I'|-1 \geq k' - 1 = k$, which makes $(G,X,k)$ a \yes-instance.
\end{proof}

\begin{redrule} \label{rr:ccRemoved}
    If there exists a connected component $P$ in $F$ which is not a non-redundant anchor triangle, and for every chunk $X' \in \mathcal{X}$ there is no conflicts induced by $X'$ on $P$, i.e.  $\CONF_{P}(X') = 0$, then remove $P$ from $G$ and reduce $k$ by $\alpha (P)$. We let $G' := G-P$, $X' := X$ and $k' := k - \alpha (P)$.
\end{redrule}

To prove that Reduction \ref{rr:ccRemoved} is safe, we will rely on the following lemma, which states that any pseudotree has a small (at most size three) obstruction in terms of obtaining a maximum independent set.

\setcounter{customlem}{\thelemma} 
\begin{lemma} \label{lem:pseudotree-small-Z}
    Let $P$ be a pseudotree and let $Z$ be a set of vertices such that $\alpha(P) > \alpha(P-Z)$. Then there exist three (possibly non-distinct) vertices $u,v,w \in Z \cap V(P)$ such that $\alpha(P) > \alpha(P-\{u,v,w\})$.
\end{lemma}

The proof of the above lemma is quite technical, and is postponed till Section~\ref{sec:pseudotree-small-Z} in order to preserve the flow of the kernelization algorithm. Taking Lemma~\ref{lem:pseudotree-small-Z} as a black box, we are able to make the following observation:

\begin{observation} \label{obs:small-conflict}
    Let $P \subseteq F$ be a connected component in the pseudoforest $F$ and let $X' \subseteq X$ be an independent set such that $\CONF_{P}(X') > 0$. Then there exists some $X'' \subseteq X'$ with $1 \leq |X''| \leq 3$ such that $\CONF_{P}(X'') > 0$.
\end{observation}

We see that the observation is true, since by Lemma~\ref{lem:pseudotree-small-Z} there exist $u,v,w \in N_G(X') \cap V(P)$ such that $\alpha(P) > \alpha(P-\{u,v,w\})$. Then for each element $u,v,w$, pick an arbitrary neighbor $x_u,x_v,x_w \in X'$ (they need not be distinct) to form the set $X'' := \{x_u, x_v, x_w\}$. See that then $\CONF_{P}(X'') > 0$. We are now equipped to prove safeness of Reduction~\ref{rr:ccRemoved}.

\begin{lemma} \label{lem:ccRemoved}
    Reduction~\ref{rr:ccRemoved} is safe. Let $(G, X, k)$ be an instance of \ISPFshort\ to which Reduction~\ref{rr:ccRemoved} is applicable, and let $(G', X', k')$ be the reduced instance. Then $(G, X, k)$ is a \yes-instance if and only if $(G', X', k')$ is a \yes-instance.
\end{lemma}

\begin{proof}
    Let $P \subseteq F$ be the connected component which triggered the reduction.
    
    For the forward direction of the proof, assume $(G, X, k)$ is a \yes-instance and let $I$ be an independent set of $G$ with size at least $k$. Let $I' := I \setminus V(P)$. Clearly $I'$ is an independent set of $G'$. Now observe that $|I \cap V(P)| \leq \alpha (P)$, and thus $|I'| = |I| - |I \cap V(P)| \geq k - \alpha (P) = k'$. Hence $(G', X', k')$ is a \yes-instance.
    
    For the backward direction, we assume that $(G', X', k')$ is a \yes-instance, and has an independent set $I'$ of size at least $k'$. Because of Lemma~\ref{lem:unnecessary} we can assume that $I'$ contains no unnecessary triples ${_3X} \subseteq X' \cap I'$. We want to show that we can always pick some independent set $I_P \subseteq V(P)$ with $|I_P| = \alpha(P)$ such that $I := I' \cup I_P$ is an independent set with size at least $k' + \alpha(P) = k$. Since $I'$ and $V(P)$ are disjoint in $G$ by construction, it will suffice to show that $\alpha (P-N_G(I')) \geq \alpha (P)$.
    
    
    Assume for the sake of contradiction that $\alpha (P-N_G(I')) < \alpha (P)$. Since $P$ was a connected component in $F$, all its neighbors $N_G(P)$ are in $X$. Thus we have that $\CONF_{P}(X' \cap I') > 0$. By Observation~\ref{obs:small-conflict}, we further have that there exists some $X'' \subseteq X' \cap I'$ such that $1 \leq |X''| \leq 3$ and $\CONF_{P}(X'') > 0$.
    
    For any such $X''$, there are two cases. In the first case, $\CONF_{F}(X'') < |X|$. Because $X''$ is also independent and has size at most $3$, it is a chunk of $X$ in the input instance. This contradicts the preconditions for Reduction~\ref{rr:ccRemoved}, so this case can not happen.
    
    In the second case, $\CONF_{F}(X'') \geq |X|$. But then one of Reductions \ref{rr:singlesRemoved}, \ref{rr:doublesForbidden} or \ref{rr:triplesOptional} would have previously been applied to $X''$, yielding it either unfeasible for an independent set or making it an unnecessary triple in the input instance. Because non-redundant anchor triangles are not chosen for removal by Reduction~\ref{rr:ccRemoved}, $X''$ is also an unnecessary triple in the output instance, which contradicts that $I'$ contains no unnecessary triples. This concludes the proof.
\qed\end{proof}

Notice that Reduction~\ref{rr:ccRemoved} will remove connected components from $F$. When the reduction is not applicable, we should then be able to give some bound on the number of connected components in $F$. The next lemma gives such a bound:

\begin{lemma} \label{lem:bound-cc-in-F}
    Let $(G, X, k)$ be an instance of \ISPFshort\ which is irreducible with respect to Reductions \ref{rr:singlesRemoved}, \ref{rr:doublesForbidden}, \ref{rr:triplesOptional} and \ref{rr:ccRemoved}. Let $\mathcal{C}_F$ denote the set of all connected components $P \subseteq F$. Then $|\mathcal{C}_F| \leq |X|^4 + |X|^3$, i.e.\@ the number of connected components in $F$ is at most $|X|^4 + |X|^3$.
\end{lemma}

\begin{proof}
    Let $\mathcal{A}_F$ denote the set of all non-redundant anchor triangles in $F$. Consider the bipartite graph $B$ between the chunks $\mathcal{X}$ and connected components $\mathcal{C}_F$ where there is an edge between $X' \in \mathcal{X}$ and $P \in \mathcal{C}_F$ if and only if $CONF_P(X') > 0$. Because Reduction~\ref{rr:ccRemoved} is not applicable, every connected component $P \in \mathcal{C}_F$ will have at least one edge incident to it in $B$, unless the component is a non-redundant anchor triangle. Thus any bound on the number of edges in $B$ will also be a bound for $|\mathcal{C}_F| - |\mathcal{A}_F|$. 
    
    By the definition of chunks, we know that for every $X' \in \mathcal{X}$, $CONF_F(X') < |X|$. Since conflicts induced by $X'$ on $F$ in different connected components of $F$ are distinct, each $X' \in \mathcal{X}$ is incident to less than $|X|$ connected components in $B$. Since $|\mathcal{X}| \leq |X|^3$, we get that $|\mathcal{C}_F| - |\mathcal{A}_F| \leq |X|^4$. It remains to show that $|\mathcal{A}_F| \leq |X|^3$ to conclude the proof. This can be verified by observing that every anchor triangle has a neighborhood in $X$ of size exactly 3. If two anchor triangles had the same neighborhood in $G$ (and hence in $X$), they would not be non-redundant, so there is at most one non-redundant anchor triangle for each distinct triple of $X$. Observe that the number of such distinct triples is less than $|X|^3$.
\end{proof}

When the above reduction rules have been exhaustively applied, the next two reductions will be executed exactly once each.

\begin{redrule} \label{rr:split-cycles}
    Let $\hat{X} \subseteq V(F)$ be a set such that $\hat{X}$ contains exactly one vertex of each cycle in $F$. In the reduced graph, let $G' := G$, $X' := X \cup \hat{X}$, and $k' := k$.
\end{redrule}

The reduction is safe because neither $G$ nor $k$ was changed. Observe that $X$ is now a feedback vertex set (which is fine, since every feedback vertex set is also a modulator to pseudoforest). This reduction may increase the size of $X$ dramatically. This is why the Reduction is applied only once, such that we can give guarantees for the size of the reduced instance.

\begin{observation} \label{obs:bound-X-after-cycle-split}
    Let $(G', X', k')$ be an instance of \ISPFshort\ after Reduction~\ref{rr:split-cycles} have been applied to $(G, X, k)$. Then $X' \leq |X|^4 + |X|^3 + |X|$.
\end{observation}

After Reduction~\ref{rr:split-cycles} has been applied once, the returned instance $(G,X,k)$ is ready for the final reduction step. Note that since $X$ is now a feedback vertex set, $(G,X,k)$ is now an instance of \ISFVS\ as well, and we can for the final reduction apply the kernel of Jansen and Bodlaender.

\defparproblem{\ISFVS{} (\ISFVSshort)}{A simple undirected graph $G$, a feedback vertex set $X\subseteq V(G)$, integer $k$.}{Size of the feedback vertex set $|X|$.}{Does $G$ contain an independence set of size at least $k$?}

\begin{proposition}[{\cite[Theorem~2]{JansenB11}}] \label{prop:is-by-fvs-jansen}
    \ISFVS{} has a kernel with a cubic number of vertices: There is a polynomial-time algorithm that transforms an instance $(G,X,k)$ into an equivalent instance $(G',X',k')$ such that $|X'| \leq 2|X|$, and $|V(G')| \leq 56|X|^3 + 28|X|^2 + 2|X|$.
\end{proposition}

\begin{redrule} \label{rr:reducedinstance-from-IS-by-FVS}
    Let the output instance $(G', X', k')$ be the reduced instance after applying Proposition~\ref{prop:is-by-fvs-jansen}. This reduction is applied once only.
\end{redrule}

\subsection{Bound on size of reduced instances}\label{ssec:bound}
When no reduction rules can be applied to an instance, we call it \emph{reduced}. In this section we will prove that the number of vertices in a reduced instance $(G',X',k')$ is at most $\bigO(|X|^{12})$ where $|X|$ is the size of the modulator in the original problem $(G,X,k)$. 

\begin{theorem}\label{thm:ISPF}
    \ISPF\ admits a kernel with $\bigO(|X|^{12})$ vertices. 
\end{theorem}
\begin{proof}
In order to prove the theorem, we show that where is a polynomial time algorithm that transforms an instance $(G, X, k)$ to an equivalent instance $(G', X', k')$ such that
    
    \begin{itemize}
        \item $|V(G')| \leq 56(|X|^4 + |X|^3 + |X|)^3 + 28(|X|^4 + |X|^3 + |X|)^2 + 2(|X|^4 + |X|^3 + |X|)$,
        \item $|X'| \leq 2|X|^4 + 2|X|^3 + 2|X|$, and
        \item $k' \leq k + |X|^3$.
    \end{itemize}
    We will begin with the proof that $k' \leq k + |X|^3$. The only transformation which increase $k$ is Reduction~\ref{rr:triplesOptional}, which rise $k$ by 1 each time it is applied. However, this transformation will be done less than $|X|^3$ times, since the rule will be applied at most once for each distinct triple of $X$.
    
    Next, we focus on the bound $|X'| \leq 2|X|^4 + 2|X|^3 + 2|X|$. The only transformations which increase $|X|$ are Reductions \ref{rr:split-cycles} and \ref{rr:reducedinstance-from-IS-by-FVS}, which are applied only once each. By Observation~\ref{obs:bound-X-after-cycle-split} we then have that $|X'| \leq |X|^4 + |X|^3 + |X|$ after applying Reduction~\ref{rr:split-cycles}, and by Proposition~\ref{prop:is-by-fvs-jansen} we have that the size is at most doubled after applying Reduction~\ref{rr:reducedinstance-from-IS-by-FVS}. Thus the bound holds.
    
    For the bound on $V(G)$, let us consider the instance of \ISFVSshort\ $(G'', X'', k'')$ to which Reduction~\ref{rr:reducedinstance-from-IS-by-FVS} was applied in order to obtain the final reduced instance $(G', X', k')$. We have already established that $|X''| \leq |X|^4 + |X|^3 + |X|$. It follows from Proposition~\ref{prop:is-by-fvs-jansen} that in the reduced instance, $|V(G')| \leq 2|X''| + 28|X''|^2 + 56|X''|^3$, which in terms of $|X|$ yields $|V(G')| \leq 56(|X|^4 + |X|^3 + |X|)^3 + 28(|X|^4 + |X|^3 + |X|)^2 + 2(|X|^4 + |X|^3 + |X|)$.
    
    Finally, observe that each reduction can be done in polynomial time.
\qed\end{proof}

\begin{corollary} \label{cor:VCPF}
    \VCPF\ admits a kernel with $\bigO(|X|^{12})$ vertices.
\end{corollary}

%
\section{Proof of Lemma~\ref{lem:pseudotree-small-Z}} \label{sec:pseudotree-small-Z}

%
%
%
%
%

In this section we prove Lemma~\ref{lem:pseudotree-small-Z}. Our starting point will be the following result by Jansen and Bodlaender \cite[Lemma 4]{JansenB11_c}, rephrased here in terms of a vertex set $Z$:

\begin{proposition} \label{pfm:prop:tree-small-Z}
    Let $T$ be a tree, and let $Z$ be a set of vertices. If $\alpha(T) > \alpha(T-Z)$, then there exist two (possibly non-distinct) vertices $u,v \in Z \cap V(T)$ such that $\alpha(T) > \alpha(T-\{u,v\})$.
\end{proposition}

We also need to establish a framework for reasoning about independent sets in trees which rely on whether a vertex is $\alpha$-critical or not. We will use the following definition:

\begin{definition}[$\alpha$-critical] \label{pfm:def:alphacritical}
    Let $G$ be a graph, and let $v$ be a vertex. If $v$ is in every maximum independent set of $G$, \ie $\alpha (G) = 1 + \alpha (G-v)$, then $v$ is \emph{$\alpha$-critical} in $G$.
\end{definition}

\begin{observation} \label{pfm:obs:alphacritical}
    Let $G$ be a graph and let $v$ be a vertex in $G$. Then $v$ is an $\alpha$-critical vertex of $G$ if and only if $\alpha (G-v) = \alpha (G-N[v])$.
\end{observation}

\begin{lemma} \label{pfm:lem:alphacriticalalternate}
    Let $T$ be a tree rooted at $r$. Then $r$ is $\alpha$-critical in $T$ if and only if $a$ is not $\alpha$-critical in $T_{a}$ for all children $a$ of $r$.
\end{lemma}

\begin{proof}
    For the forward direction of the proof, assume $r$ is $\alpha$-critical in $T$. Since $T$ is a tree rooted at $r$, we have that $\alpha (T-r) = \sum_{a \in C(r)} \alpha (T_a)$ and $\alpha (T-N[r]) = \sum_{a \in C(r)} \alpha (T_a - a)$.
    
    Assume for the sake of contradiction that there is some $a' \in C(r)$ such that $a'$ is $\alpha$-critical in $T_{a'}$. Then we have $\alpha(T_{a'}) = 1 + \alpha(T_{a'} - a')$, which lead us to conclude that $\sum_{a \in C(r)} \alpha (T_a) \geq 1 + \sum_{a \in C(r)} \alpha (T_a - a)$. This contradicts that $r$ is $\alpha$-critical in $T$ by Observation~\ref{pfm:obs:alphacritical}.
    
    For the backward direction of the proof, assume that for all children $a$ of $r$, $a$ is not $\alpha$-critical for $T_a$. Seeing that $T$ is a tree, we recall that $\alpha (T-r) = \sum_{a \in N(r)} \alpha(T_a)$. Adding the single vertex $r$ back will increase the independence number by at most one. It remains to show that it also increase by at least one to conclude the proof. But this can be done by picking a maximum independent set in $T-r$ which avoids all of $C(r)$, and then include $r$. Note that such a set exists by the initial assumption that no child $a$ of $r$ is $\alpha$-critical for $T_a$.
\end{proof}

\begin{figure}[ht] \label{pfm:fig:acritalternate}
    \centering
    \includegraphics[scale=1.0]{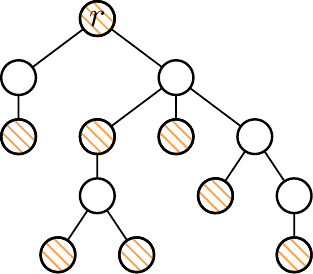}
    \caption{Lemma~\ref{pfm:lem:alphacriticalalternate}: Let $T$ be a tree rooted at $r$. A vertex $a \in V(T)$ is marked with orange stripes if $a$ it is $\alpha$-critical in the subtree $T_a$. Then $a$ is $\alpha$-critical if and only if, for all its children $b \in C(a)$, $b$ is not $\alpha$-critical in  $T_b$.}
\end{figure}

\begin{observation} \label{pfm:obs:redemption-cost}
    Let $G$ be a graph and let $a$ be an $\alpha$-critical vertex of $G$. If there is a set of vertices $Z$ such that $a$ is not $\alpha$-critical in $G-Z$, then $\alpha(G) > \alpha(G-Z)$.
\end{observation}

\begin{observation} \label{pfm:obs:alpha-of-trees-avoiding-Z}
    Let $T$ be a tree rooted at $r$, and let $Z$ be a set of vertices. Then the following holds:
    
    \begin{itemize}
        \item If $r$ is $\alpha$-critical in $T-Z$, then $\alpha(T-Z) = 1 + \sum_{a \in C(r)} \alpha(T_{a} - Z)$. In particular, if $r$ is $\alpha$-critical in $T$ then $\alpha(T) = 1 + \sum_{a \in C(r)} \alpha(T_{a})$.
        \item If $r$ is not $\alpha$-critical in $T-Z$, then $\alpha(T-Z) = \sum_{a \in C(r)} \alpha(T_{a} - Z)$. In particular, if $r$ is not $\alpha$-critical in $T$ then $\alpha(T) = \sum_{a \in C(r)} \alpha(T_{a})$.
    \end{itemize}
\end{observation}

Armed with Proposition~\ref{pfm:prop:tree-small-Z} and the framework for reasoning about independent sets in trees presented above, we are now prepared to show the next lemma:

\begin{lemma} \label{pfm:lem:alphacrit-trees-inductive}
    Let $T$ be a tree rooted at $r$, and let $Z$ be a set of vertices. Then the following holds:
    \begin{enumerate}[label=(\alph*)]
        \item \label{pfm:cse:crit} If $\alpha (T) > \alpha (T-Z)$ and $r$ is $\alpha$-critical in $T$, then either
        \begin{enumerate}[label=(\roman*)]
            \item There exist two (possibly non-distinct) vertices $u,v \in Z \cap V(T)$ such that $\alpha(T) > \alpha(T-\{u,v\})$ and $r$ is still $\alpha$-critical in $T-\{u,v\}$.
            \item There exists a vertex $u \in Z \cap V(T)$ such that $\alpha(T) > \alpha(T-u)$ and $r$ is not $\alpha$-critical in $T - u$.
        \end{enumerate}
        \item \label{pfm:cse:noncrit} If $\alpha(T) = \alpha(T-Z)$ and $r$ is $\alpha$-critical in $T-Z$ but not in $T$, then there exists a vertex $u \in Z \cap V(T)$ such that $r$ is $\alpha$-critical in $T-u$.
    \end{enumerate}
\end{lemma}

\begin{proof}
    The proof is done by strong induction on the height of the tree, where we will go two levels deep. For the base case we consider trees of height $0$ and $1$, \ie the single vertex and stars rooted at the center vertex. In both cases, it is easy to verify that the lemma holds. For the inductive step, we let $T$ be a tree rooted in $r$. By the induction hypothesis we assume the lemma holds for all strict subtrees of $T$, as these have all strictly smaller height than $T$.  We will now show each part of the lemma separately.
    
    \textbf{Part~\ref{pfm:cse:crit}.} We assume that $\alpha(T) > \alpha(T-Z)$, and that $r$ is $\alpha$-critical in $T$. First consider the case when $r \in Z$. Then just pick $u := r$ to meet the requirement for (ii). Thus, for the remainder of the case we can assume $r \notin Z$. By Lemma~\ref{pfm:lem:alphacriticalalternate}, we can also assume that there is no child $a$ of $r$ such that $a$ is $\alpha$-critical in $T_a$. We will now consider two cases.
    
    \textit{Case 1.} This case applies if there is no child $a$ of $r$ such that $\alpha(T_{a}) > \alpha(T_{a}-Z)$. Then there must exist at least one child $a'$ of $r$ such that $a'$ is $\alpha$-critical in $T_{a'} - Z$, or else there is a contradiction with the assumption that $\alpha(T) > \alpha(T-Z)$. Hence we can apply the induction hypothesis \ref{pfm:cse:noncrit} to $T_{a'}$ and find that there exists a vertex $u \in Z \cap V(T_{a'})$ such that $a'$ is $\alpha$-critical in $T_{a'} - u$. Observe that by Lemma~\ref{pfm:lem:alphacriticalalternate}, $r$ is not $\alpha$-critical in $T-u$. Further applying Observation~\ref{pfm:obs:alpha-of-trees-avoiding-Z} to find that $\alpha(T) > \alpha(T-u)$, we see that we have met the requirement of (ii).
    
    \textit{Case 2.} This case applies if there exists a child $a$ of $r$ such that $\alpha(T_a) > \alpha(T_a - Z)$. Before we proceed further, we want to establish that there must exist some child $b'$ of $a$ such that $\alpha(T_{b'}) > \alpha(T_{b'} - Z)$. For the sake of contradiction, assume there is not, \ie $\sum_{b \in C(a)} \alpha(T_{b}) = \sum_{b \in C(a)} \alpha(T_{b} - Z)$. Because $a$ is not $\alpha$-critical in $T_a$ by Lemma~\ref{pfm:lem:alphacriticalalternate}, we have that $\alpha(T_a) = \sum_{b \in C(a)} \alpha(T_{b})$ by Observation~\ref{pfm:obs:alpha-of-trees-avoiding-Z}. Now note that $\sum_{b \in C(a)} \alpha(T_{b} - Z) \leq \alpha(T_a-Z)$. Combining the above, we obtain $\alpha(T_a) \leq \alpha(T_a-Z)$, which contradicts the initial assumption of this case. Thus for the remainder of the case we can assume there is at least one child $b'$ of $a$ such that $\alpha(T_{b'}) > \alpha(T_{b'} - Z)$. By Lemma~\ref{pfm:lem:alphacriticalalternate}, we also have that there exists at least one child $b''$ of $a$ which is $\alpha$-critical in $T_{b''}$. Now we will again distinguish between two different cases:
    
    \textit{Subcase 2.1.} This subcase applies if there exist two distinct children $b$ and $b'$ of $a$ such that $\alpha(T_b) > \alpha(T_b-Z)$ and $b'$ is $\alpha$-critical in $T_{b'}$. 
	By Lemma~\ref{pfm:prop:tree-small-Z}, there exist $u,v \in Z \cap V(T_b)$ such that $\alpha(T_b) > \alpha(T_b-\{u,v\})$. Note that $a$ is not $\alpha$-critical in $T_a - \{u,v\}$ by Lemma~\ref{pfm:lem:alphacriticalalternate} because $b'$ is $\alpha$-critical in $T_{b'} - \{u,v\}$ ($T_{b'}$ is untouched by $u,v$). In order to meet the requirements of (i), we will proceed to show that $\alpha(T_a) > \alpha(T_a - \{u,v\})$.

\begin{figure}[b]
	\label{fig:inductivepf}
    \centering
    \begin{minipage}{.48\linewidth}
        \centering
        \includegraphics[scale=.9]{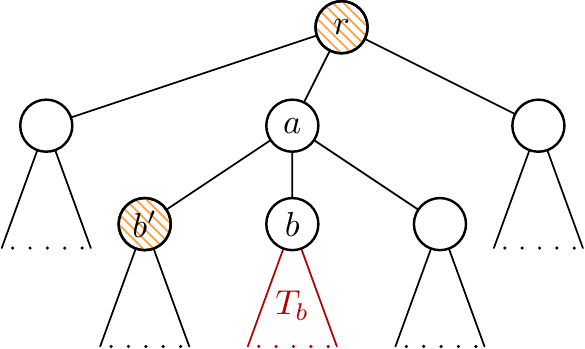}
        Subcase~2.1
    \end{minipage}%
    \hspace*{.04\linewidth}
    \begin{minipage}{.48\linewidth}
        \centering
        \includegraphics[scale=.9]{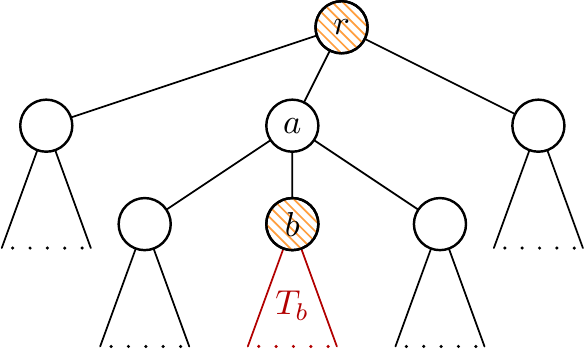}
        Subcase~2.2
    \end{minipage}

    \caption{Proof of Lemma~\ref{pfm:lem:alphacrit-trees-inductive} (a), Case 2: Orange stripes indicates that the vertex is $\alpha$-critical in the subtree rooted at that vertex. Subtrees where $\alpha(T_a) > \alpha(T_a-Z)$ are drawn with a red color.}
\end{figure}
	
    To show this, we use Observation~\ref{pfm:obs:alpha-of-trees-avoiding-Z} to find that $\alpha(T_a) = \sum_{b'' \in C(a)} \alpha(T_{b''})$ and $\alpha(T_a - \{u,v\}) = \sum_{b'' \in C(a)} \alpha(T_{b''}-\{u,v\})$. Because there exists a child $b$ of $a$ such that $\alpha(T_b) > \alpha(T_b-\{u,v\})$, it then follows that $\alpha(T_a) > \alpha(T_a - \{u,v\})$. Because $a$ is not $\alpha$-critical in $T_a - \{u,v\}$ and the siblings of $a$ are untouched by $u,v$, it follows that $r$ is $\alpha$-critical by Lemma~\ref{pfm:lem:alphacriticalalternate}. By symmetric reasoning as above but using $T$ and $r$ rather than $T_a$ and $a$, we conclude that also $\alpha(T) > \alpha(T-\{u,v\})$, thus obtaining the requirements for (i).
    
    \textit{Subcase 2.2.} This subcase applies if there is exactly one child $b$ of $a$ which is $\alpha$-critical in $T_b$, and this is the only child of $a$ for which $\alpha(T_b) > \alpha(T_b - Z)$. 
	The vertex $a$ may have other children besides $b$, but for all of these children $b' \in C(a) \setminus \{b\}$, we have that $b'$ is not $\alpha$-critical in $T_{b'}$ and $\alpha(T_{b'}) = \alpha(T_{b'} - Z)$. By the induction hypothesis \ref{pfm:cse:crit} there are two cases:

\begin{itemize}
    \item There exist two (possibly non-distinct) vertices $u,v \in Z \cap V(T_b)$ such that $\alpha(T_b) > \alpha(T_b-\{u,v\})$ and $b$ is $\alpha$-critical in $T_b - \{u,v\}$. This case leaves $a$ not $\alpha$-critical in $T_a - \{u,v\}$, and yields $\alpha(T_a) > \alpha(T_a-\{u,v\})$. By an analogous argument to that of the previous subcase, we obtain the requirements for (i).
    \item There exists $u \in Z \cap V(T_b)$ such that $\alpha(T_b) > \alpha(T_b-u)$, and such that $b$ is not $\alpha$-critical in $T_b - u$. In this case, $a$ is $\alpha$-critical in $T_a - u$. This, however, in turn yields $r$ not $\alpha$-critical in $T - u$, which must reduce the independence number of $T$ by Observation~\ref{pfm:obs:redemption-cost}. Thus, we have obtained the requirement for (ii).
\end{itemize}

    We have now exhausted all possibilities, and in each case obtained the requirements for either (i) or (ii). This concludes the proof for part~\ref{pfm:cse:crit}.
    
    \textbf{Part~\ref{pfm:cse:noncrit}.} We assume that $\alpha(T) = \alpha(T-Z)$ and that $r$ is $\alpha$-critical in $T-Z$, but not in $T$. By Lemma~\ref{pfm:lem:alphacriticalalternate}, we know that $r$ is $\alpha$-critical in $T-Z$ if and only if no child $a$ of $r$ is $\alpha$-critical in $T_a - Z$. If there are two or more children which are $\alpha$-critical in their respective subtrees before forbidding $Z$, we note that $\alpha(T) > \alpha(T-Z)$ by Observations \ref{pfm:obs:redemption-cost} and \ref{pfm:obs:alpha-of-trees-avoiding-Z}, contradicting the premise of part (b) in the lemma. Thus there is exactly one child $a$ of $r$ which is $\alpha$-critical in $T_a$, and for all other children $a'$ of $r$ it holds that $a'$ is not $\alpha$-critical in neither $T_{a'}$ nor $T_{a'} - Z$, as this would contradict either that $\alpha(T) = \alpha(T-Z)$ or that $r$ is $\alpha$-critical in $T-Z$.
    
    By Observation~\ref{pfm:obs:redemption-cost}, we know that $\alpha(T_a) > \alpha(T_a - Z)$. By the induction hypothesis (a), there are two possibilities; either (i) is true for $T_a$, and there exist vertices $u,v \in Z$ such that $\alpha(T_a) > \alpha(T_a - \{u,v\})$ and so that $a$ is $\alpha$-critical in $T_a - \{u,v\}$. That, however, would contradict the premise that $\alpha(T) = \alpha(T-Z)$. We may then assume that only (ii) is true, and that there exists a vertex $u \in Z \cap V(T_a)$ such that $\alpha(T_a) > \alpha(T_a - u)$ and which leave $a$ not $\alpha$-critical in $T_a - u$. See that this choice of $u$ will also leave $r$ $\alpha$-critical in $T-u$, concluding the proof for part (b).
\end{proof}

We will now move on to pseudotrees, \ie graphs which contains at most one cycle, and study how maximum independent sets behave in them when a set $Z$ of vertices forbidden to be included in any independent set is introduced. We can view a pseudotree $P$ as a cycle $C$ with rooted trees being attached to vertices of $C$ (if there is no cycle, $P$ is a tree and behaves accordingly). A vertex in $C$ may have zero, one or several trees being attached to it. Every non-cycle vertex with a neighbor in $C$ is the root $r$ of a tree $T$ attached to the cycle. Formally, we use this definition:

\begin{definition}[Attached tree]
    Let $P$ be a pseudotree which contains a cycle $C \subseteq P$. Each connected component $T$ of $P-C$ is then an \emph{attached tree} of $C$. Further, let $c \in V(C)$ and $r \in V(T)$ be the two vertices such that $cr \in E(P)$ (there is a unique pair with this property since $P$ is a pseudotree). Then $r$ is designated as the root of $T$, and $c$ is the \emph{attachment point} vertex for $T$. We say that $c$ has an attached tree $T$, and that $T$ is \emph{attached} to $C$ at $c$.
\end{definition}

\begin{figure}[ht] \label{pfm:fig:attached}
    \centering
    \includegraphics[scale=1.0]{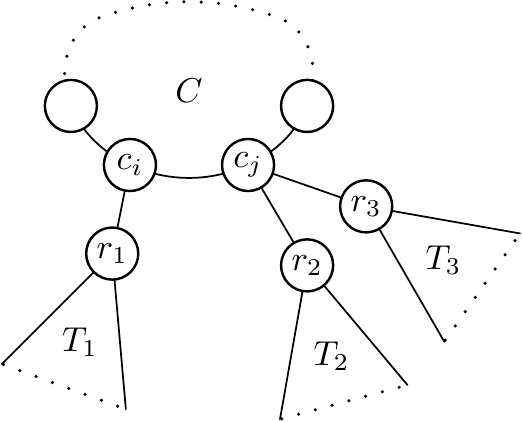}
    \caption{Attached tree: A pseudotree consists of a cycle $C$ and attached subtrees. We say that the attached tree $T_1$ is rooted at $r_1$, and is attached to $C$ at its attachment point vertex $c_i$. The cycle vertex $c_j$ has two attached trees, $T_2$ and $T_3$.}
\end{figure}

In building a maximum independent set $I$ for a pseudotree $P$, each attached tree is faced with two possibilities; either their attachment point vertex $c \in C$ is in $I$, or it is not. In some cases, $c$ being in $I$ may reduce the number of vertices available to pick in a tree $T$ attached to $c$, since the root $r$ of $T$ is forbidden in $I$ due to it being a neighbor of $c$. In cases like this, we observe that it is always at least as good to let $c$ be left outside $I$ and rather pick the larger independent set from $T$. This motivates distinguishing between trees where the root is $\alpha$-critical in the tree from those where it is not.

With this in mind, we observe that the following greedy strategy for finding a MIS in a pseudotree $P$ is correct: First, pick a MIS in every tree $T$ attached to $C$ which avoids the root $r$ of $T$ unless $r$ is $\alpha$-critical in $T$. Next, find a MIS for the cycle $C$ which avoids any $c \in C$ which has an attached tree $T$ where the root $r$ is $\alpha$-critical in $T$. The union of found independent sets will be a MIS for $P$. We can easily extend this algorithm to also avoid a set $Z$ of vertices:

\begin{algoss}[Greedy MIS for pseudotree avoiding Z] \label{pfm:alg:grdy-MIS-noZ-pt}
    Input: A pseudotree $P$ and a set of vertices $Z$.
    Output: A maximum independent set $I$ of $P-Z$.
    \begin{enumerate}
        \item If $P$ contain no cycle $C$, then return the MIS of $P$ avoiding $Z$.
        \item Let $Z' := \emptyset$. This set is for cycle vertices which will be marked as forbidden due to $\alpha$-criticalness of roots in attached subtrees.
        \item For every tree $T_0, T_1, \cdots, T_t$ attached to the cycle $C$, let $r_i$ be the root of $T_i$ and let $c_i \in C$ be the attachment point vertex for $T_i$. Then for $i \in \{1,2,\cdots,t\}$, do the following:
        \begin{enumerate}
            \item Let $I_{T_i}$ be a MIS for $T_i-Z$ which avoids $r_i$ unless $r_i$ is $\alpha$-critical in $T_i-Z$.
            \item If $r_i$ is $\alpha$-critical in $T_i-Z$, mark $c_i$ as forbidden: Let $Z' := Z' \cup \{c_i\}$.
        \end{enumerate}
        \item Let $I_{C}$ be a MIS for the cycle vertices $C - (Z\cup Z')$.
        \item Let the final solution be the union of the found sets: Return  $(\bigcup_{i=1}^t I_{T_i}) \cup I_{C}$.
    \end{enumerate}
\end{algoss}

When Algorithm~\ref{pfm:alg:grdy-MIS-noZ-pt} is called on input $(P,Z)$, we refer to sets of this solution with a superscript $^{(P,Z)}$, \eg the set $I_{T_i}^{(P,Z)}$ refers to the set $I_{T_i}$ when the algorithm is called on input $(P,Z)$. However, note that we will never actually call this algorithm during the kernelization of \ISPFshort{}, it is used only for analysis. Seeing that the above algorithm is correct, we are now ready to prove the final lemma of this section:

\begin{customlem}
    Let $P$ be a pseudotree and let $Z$ be a set of vertices such that $\alpha(P) > \alpha(P-Z)$. Then there exist three (possibly non-distinct) vertices $u,v,w \in Z \cap V(P)$ such that $\alpha(P) > \alpha(P-\{u,v,w\})$.
\end{customlem}

\begin{proof}
    Assume the condition of the lemma holds. If $P$ has no cycle, then by Lemma~\ref{pfm:prop:tree-small-Z} there exist $u,v \in Z \cap V(P)$ such that $\alpha(P) > \alpha(P-\{u,v\})$, and we have obtained the requirement of the lemma (we simply let one of the elements repeat, e.g. let $w:=v$). For the remainder of the proof, we may therefore assume that $P$ has a cycle $C$.
    
    If there is a vertex $w$ of the cycle $C$ in $P$ which is also in $Z$, then pick it as one of the three elements. Unless $\alpha(P) > \alpha(P-w)$, we have by Lemma~\ref{pfm:prop:tree-small-Z} that there are two vertices $u,v \in Z \cap V(P)$ such that $\alpha(P-w) > \alpha(P-\{u,v,w\})$. In either case we have obtained the requirement of the lemma. For the remainder of the proof we may therefore assume that there is no vertex of the cycle $C \subseteq P$ which is also in $Z$.
    
    Let us next consider the trees attached to $C$. If there exists a tree $T$ rooted at $r$ which is of one of the following types, then we can find $u,v,w \in Z$ which satisfy the requirement of the lemma:
    
    \begin{itemize}
        \item \emph{$r$ is not $\alpha$-critical in $T$ and $\alpha(T) > \alpha(T-Z)$.} Then by Lemma~\ref{pfm:prop:tree-small-Z}, there exist $u,v \in Z \cap V(T)$ such that $\alpha(T) > \alpha(T-\{u,v\})$. Observe that Algorithm~\ref{pfm:alg:grdy-MIS-noZ-pt} will produce sets $|I_{T}^{(P, \emptyset)}| > |I_{T}^{(P, \{u,v\})}|$ and $|I_{C}^{(P, \emptyset)}| \geq |I_{C}^{(P, \{u,v\})}|$, so by the correctness of the algorithm we conclude that $\alpha(P) > \alpha(P-\{u,v\})$. The requirement of the lemma is thus satisfied.
        
        \item \emph{$r$ is $\alpha$-critical in $T$, $\alpha(T) > \alpha(T-Z)$, and case (i) of Lemma~\ref{pfm:lem:alphacrit-trees-inductive}~\ref{pfm:cse:crit} holds for $T,r,Z$.} Then there exist $u,v \in Z\cap V(T)$ such that $\alpha(T) > \alpha(T-\{u,v\})$ and $r$ is $\alpha$-critical in $T-\{u,v\}$. In order to show that $\alpha(P) > \alpha(P-\{u,v\})$, consider what Algorithm~\ref{pfm:alg:grdy-MIS-noZ-pt} will do: Whether called on $(P, \emptyset)$ or $(P, \{u,v\})$, the attachment point vertex $c \in C$ for $T$ will still be marked as forbidden. Thus the only difference occurs in $I_{T}$, which is strictly smaller in the latter case. Hence, the requirement of the lemma is satisfied.
        
        \item \emph{$r$ is $\alpha$-critical in $T$, and $\alpha(T) = \alpha(T-Z)$.} Let $c$ be the attachment point vertex for $T$. Observe that $\alpha((P-T)-c) > \alpha((P-T)-(\{c\}\cup Z))$, or else there is a contradiction with the preconditions of the lemma. Since $(P-T)-c$ is a tree, we then have by Lemma~\ref{pfm:prop:tree-small-Z} that there exist $u,v \in Z \cap V((P-T)-c)$ such that $\alpha((P-T)-c) > \alpha((P - T) - \{c,u,v\})$. Observe that also $\alpha(P) > \alpha(P-\{u,v\})$, and the lemma is satisfied.
    \end{itemize}
    
    From now on, we can assume that there are no trees attached to $C$ of the above types. We observe that every tree $T$ which is attached to $C$ with root $r$ must then be one of these three types instead:
    \begin{enumerate}[label=(\alph*)]
        \item $r$ is $\alpha$-critical in $T$, and $\alpha(T) > \alpha(T-Z)$. Case (i) of Lemma~\ref{pfm:lem:alphacrit-trees-inductive}~\ref{pfm:cse:crit} does not apply, so by case (ii), there exists a singleton $u \in Z\cap V(T)$ such that $\alpha(T) > \alpha(T-u)$ and $r$ is not $\alpha$-critical in $T-u$.
        \item $r$ is $\alpha$-critical in $T-Z$, but not in $T$. We know that $\alpha(T) = \alpha(T-Z)$, and by Lemma~\ref{pfm:lem:alphacrit-trees-inductive}~\ref{pfm:cse:noncrit} that there exists a singleton $u \in Z \cap V(T)$ such that $r$ is $\alpha$-critical in $T-u$.
        \item $r$ is $\alpha$-critical in neither $T$ nor $T-Z$. We know that $\alpha(T) = \alpha(T-Z)$. 
    \end{enumerate}
    
    If there is a vertex of the cycle $c\in C$ such that two trees $T_a$ and $T_a'$ both of type (a) are attached to $c$ with respective roots $r$ and $r'$, then we have that there exists $u \in Z \cap V(T_a)$ such that $\alpha(T_a) > \alpha(T_a-u)$. Note that $u$ is sufficient to yield $\alpha(P) > \alpha(P-u)$. To see this, consider what happens in Algorithm~\ref{pfm:alg:grdy-MIS-noZ-pt}: Regardless of whether it was called with parameters $(P, \emptyset)$ or $(P, \{u\})$, $c$ will be marked as forbidden since $r'$ is still $\alpha$-critical in $T_a'-u$ ($u$ is not in $T_a'$). The only difference between the two runs of the algorithm is that $I_{T_a}$ will be strictly smaller in the latter case. Hence the requirement of the lemma holds, and going forward we can assume no such cycle vertex exists.
    
    If there is a cycle vertex $c$ with two attached trees $T_a$ and $T_b$ of types (a) and (b) respectively, then there are elements $u\in Z \cap V(T_a)$ and $v \in Z \cap V(T_b)$ such that $\alpha(T_a) > \alpha(T_a-u)$ and $r_{b}$ is $\alpha$-critical in $T_{b}-v$. Note that then $\alpha(P) > \alpha(P-\{u,v\})$. To see this, consider what happens in Algorithm~\ref{pfm:alg:grdy-MIS-noZ-pt}: Regardless of whether it was called with input $(P, \emptyset)$ or $(P, \{u,v\})$, $c$ will be marked as forbidden because $r_a$ is $\alpha$-critical in $T$, and $r_{b}$ is $\alpha$-critical in $T_{b} - v$. Since $I_{T_a}^{(P,\{u,v\})}$ is strictly smaller than $I_{T_a}^{(P,\emptyset)}$, we have reached the requirement of the lemma. From now on we assume there are no such cycle vertices.
    
    We can now partition all the cycle vertices $c \in C$ into the following three categories:
    \begin{description}
        \item[Redeemable] A cycle vertex $c\in C$ is \emph{redeemable} if it has exactly one attached tree $T$ of type (a), and all other attached trees are of type (c). We know there exists a singleton $u \in Z \cap V(T)$ such that $c$ is \emph{redeemed}, \ie such that $T$ will have its root $r$ not $\alpha$-critical in $T-u$. Informally we may note that the price payed in $T$ for redeeming $c$ is at least one by Observation~\ref{pfm:obs:redemption-cost}, \ie Algorithm~\ref{pfm:alg:grdy-MIS-noZ-pt} yields $|I_T^{(P,\emptyset)}| > I_T^{(P,\{u\})}$.
        \item[Blockable] A cycle vertex $c\in C$ is \emph{blockable} if it has at least one attached tree $T$ of type (b). It may also have any number of attached trees of type (c). We know there exists a singleton $u \in Z \cap V(T)$ such that $c$ is \emph{blocked}, \ie such that the root $r$ of $T$ will be $\alpha$-critical in $T-u$.
        \item[Free] A cycle vertex $c\in C$ is \emph{free} if all attached trees are of type (c).
    \end{description}
    
    If there are no redeemable cycle vertices in $C$, then consider some blockable cycle vertex $c_u$. It exists, or else we have that $\alpha(P) = \alpha(P-Z)$. Because there are no trees where $Z$ cause a drop in the independence number of the tree, the difference must occur in $|I_{C}|$ when using Algorithm~\ref{pfm:alg:grdy-MIS-noZ-pt} on the inputs $(P, \emptyset)$ and $(P, Z)$. By Lemma~\ref{pfm:prop:tree-small-Z}, we can extrapolate that there exist two (possibly non-distinct) blockable cycle vertices $c_v,c_w \in C-c_u$ such that $\alpha(C-c_u) > \alpha(C-\{c_u, c_v, c_w\})$. If we now let $u,v,w$ be elements that block their respective cycle vertex $c_u,c_v,c_w$, then we have that $\alpha(P) > \alpha(P-\{u,v,w\})$, and the lemma is satisfied. Going forth, we will assume there is at least one redeemable cycle vertex.
    
    If there is exactly one redeemable cycle vertex $c$, observe that Algorithm~\ref{pfm:alg:grdy-MIS-noZ-pt} will pick $\alpha(C)$ vertices of $C$ when finding the MIS for input $(P,\emptyset)$. Let $T$ be the tree of type (a) attached to $c$. See that there is no way to compensate the cost occurring in $I_T$ when $u$ redeems $c$. Thus, $\alpha(P) > \alpha(P-u)$, and the lemma holds. Assume then, there are at least two redeemable cycle vertices.
    
    If there are exactly two redeemable cycle vertices $c_1$ and $c_2$, then similarly to the previous case, Algorithm~\ref{pfm:alg:grdy-MIS-noZ-pt} will pick at least $\alpha(C)-1$ vertices of $C$ on input $(P,\emptyset)$. If we pick $u,v$ such that they redeem $c_1$ and $c_2$ respectively, then the price payed is at least two, out of which at most one can be compensated in $C$. Thus, $\alpha(P) > \alpha(P-\{u,v\})$, and the lemma holds. Assume then, there are at least three redeemable cycle vertices.
    
    Pick two redeemable vertices $c_i$ and $c_j$ and a maximal length path along the cycle $Q = \{c_1, c_2, \cdots, c_{i-1}, c_i, c_{i+1}, \cdots, c_{j-1}, c_j, c_{j+1}, \cdots, c_q\}$ such that $c_i$ and $c_j$ are the only redeemable vertices in $Q$. Since there are at least three redeemable vertices in $C$, such a path exists, and $c_1$ is not a neighbor of $c_q$. Let $\alpha(Q)$ be the maximum independent set of $Q$ after $c_i$ and $c_j$ have been both redeemed. It remains to observe that $\alpha(Q) \leq 1 + \alpha(Q-\{c_i, c_j\})$ to conclude that redeeming $c_i$ by $u$ and $c_j$ by $v$ will cause $\alpha(P) > \alpha(P-\{u,v\})$. This concludes the proof of the lemma.

\end{proof}

\section{No polynomial kernel for \VCMFshort{}}\label{sec:lowb}

In this section we show that \VCMF{} admits no polynomial kernel unless \NPincoNPpoly{}. Our strategy is to make a reduction from \CNF{} parameterized by the number of variables to \ISMFshort{}. By the immidiate correspondance between \VCproblem{} and \ISproblem{}, the result for \VCMF{} will follow. We define the following problem.

\defparproblem{\ISMF{} (\ISMFshort)}{A simple undirected graph $G$, a mock forest modulator $X\subseteq V(G)$ such that no two cycles of $G-X$ share a vertex, and an integer $k$.}{Size of a mock forest modulator $|X|$.}{Does $G$ contain an independent set of size at least $k$?}

Our reduction also shows that there is no polynomial kernel for \VCproblem{} when parameterized by the size of modulators to cactus graphs and outerplanar graphs as well, under the same condition. This strategy is an adaption of the strategy used by Jansen, Raman and Vatshelle~\cite{JansenRV14} to show that \FVSproblem{} does not admit a polynomial kernel parameterized by a modulator to mock forests unless \NPincoNPpoly{}.
%

\begin{definition}[Polynomial-parameter transformation {\cite{BodlaenderTY11}}]
    Let $\Q, \Q' \subseteq \Sigma ^* \times \N$ be parameterized problems. A polynomial-parameter transformation from $Q$ to $Q'$ is an algorithm that, on input $(x,k) \in \Sigma ^* \times \N$, takes time polynomial in $|x| + k$, and outputs an instance $(x,k) \in \Sigma ^* \times \N$ such that 
        $k'$ is polynomially bounded in $k$, and
        $(x,k) \in \Q$ if and only if $(x',k') \in \Q'$.
    %
    For a parameterized problem $\Q \subseteq \Sigma^* \times \N$, the unparameterized version of $\Q$ is the set $\hat{\Q} = \{x1^k \mid (x,k) \in \Q\}$ where 1 is a new symbol that is added to the alphabet.
\end{definition}

\begin{proposition}[\cite{BodlaenderTY11}] \label{ismf:prop:poly-par-trans}
    Let $\Q$ and $\Q'$ be parameterized problems and let $\hat{\Q}$ and $\hat{\Q}'$ be the unparameterized versions of $\Q$ and $\Q'$ respectively. Suppose $\hat{\Q}$ is NP-hard and $\hat{\Q}'$ is in NP. If there is a polynomial-parameter transformation form $\Q$ to $\Q'$, and $\Q'$ has a polynomial kernel, then $\Q$ also has a polynomial kernel.
\end{proposition}

\begin{proposition}[\cite{DellM10,FortnowS11}] \label{ismf:prop:cnf-sat-hard}
    \CNF{} parameterized by the number of variables does not admit a polynomial kernel unless \NPincoNPpoly{}.
\end{proposition}

\begin{definition}[Clause gadget]
    Let $k \geq 1$ be an integer. The \emph{clause gadget of size k} is the graph $\G_k$ consisting of $k$ triangles $T_1, T_2, \ldots , T_k$ and two extra vertices $r_0$ and $l_{k+1}$ connected as follows: For each triangle $T_i$, label the three vertices $l_i$,$r_i$, and $s_i$ (\emph{left vertex}, \emph{right vertex} and \emph{spike vertex}, respectively). Then for each $i \in \{0\} \cup [k]$, let there be an edge $r_il_{i+1}$ connecting the right vertex of $T_i$ to the left vertex of $T_{i+1}$. In this way, $\G_k$ is a ``path'' of $k$ connected triangles, with two extra degree-1 vertices attached at the ends.
\end{definition}
\vspace{-.5cm}
\begin{figure}[h]
    \centering
    \includegraphics[scale=0.9]{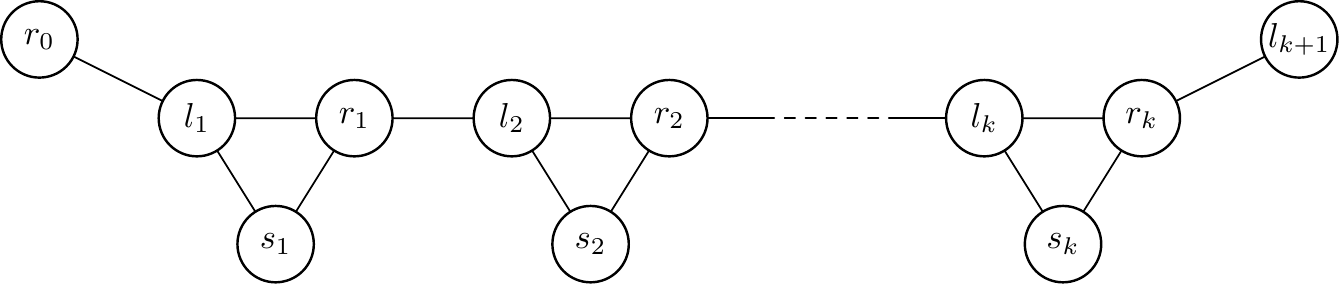}
    \caption{A clause gadget $\G_k$}
\end{figure}
\vspace{-1cm}
\begin{observation} \label{ismf:obs:gadget-mis-size}
    For a clause gadget $\G_k$, the independence number $\alpha(\G_k)$ is exactly $k+2$. This can be obtained by the independent set containing all the spike vertices as well as $r_0$ and $l_{k+1}$. We verify that this is also optimal since at most one vertex can be chosen from each triangle $T_i$, and there are only two non-triangle vertices.
\end{observation}

\begin{observation} \label{ismf:obs:at-least-one-spike-mis}
    For a clause gadget $\G_k$, every maximum independent set $I$ must contain at least one spike vertex. Removing the spike vertices, what remains of $\G_k$ is an even path with $2k + 2$ vertices, yielding a maximum independent set of size $k+1$, which is strictly smaller than $\alpha(\G_k)$.
\end{observation}

\begin{observation} \label{ismf:obs:one-spike-mis}
    For a clause gadget $\G_k$, let $S$ denote the set of spike vertices. Observe that for each spike vertex $s_i \in S$, there exists a maximum independent set $I_i$ such that $s_i$ is the only spike in $I_i$, \ie{} $I_i \cap S = \{s_i\}$.
\end{observation}

\begin{theorem}\label{thm:mockforest}
    \ISMF{} does not admit a polynomial kernel unless \NPincoNPpoly{}.
\end{theorem}

\begin{proof}
    Since \VCproblem{} is in NP and \CNF{} is NP-hard, we have by Propositions~\ref{ismf:prop:poly-par-trans} and \ref{ismf:prop:cnf-sat-hard},  that it is sufficient to show a polynomial-parameter transformation from \CNF{} parameterized by the number of variables to \ISMFshort{}.
    
    Consider an instance $F$ of \CNF{} consisting of clauses $C_1, C_2, \ldots , C_m$ over the variables $x_1, x_2, \ldots , x_n$. For a clause $C_j$, let $h(j)$ denote the number of literals in $C_j$. We will in polynomial time construct an instance $(G, X, k)$ of \ISMFshort{} such that $F$ is satisfied if and only if $(G, X, k)$ is a \yes-instance. We construct  graph $G$ as follows.
    
    For each variable $x_i$, we let there be two vertices $t_i$ and $f_i$ in $V(G)$. Let them be connected by an edge $t_if_i \in E(G)$. Which of $t_i$ and $f_i$ is included in a maximum independent set for $G$ will represent whether the variable $x_i$ is set to true or false. 
    
    For each clause $C_j$, let $\ell_1, \ell_2, \ldots , \ell_{h(j)}$ denote the literals of $C_j$. Let $\C_j$ be a copy of the clause gadget $\G_{h(j)}$, and add it to the graph $G$. Let the spikes of $\C_j$ be denoted $s_{1}, s_{2}, \ldots , s_{h(j)}$. We will connect $\C_j$ to the rest of $G$ as follows: For each literal $\ell_r \in C_j$, let there be an edge from $s_{r}$ to $f_i$ if and only if $\ell_r = x_i$. Similarly, let there be an edge from $s_{r}$ to $t_i$ if and only if $\ell_r = \neg x_i$. By this process, every spike of $\C_j$ is connected to exactly one vertex outside of $\C_j$, which is either $t_i$ or $f_i$ for some $i \in [n]$.
    This concludes the construction of  $G$. 

Let the set $X$ consist of the variable gadget vertices, \ie{}, let $X = \{t_i \mid i\in [n]\} \cup \{f_i \mid i\in [n]\}$. Observe that $X$ is indeed a mock forest modulator for $G$, since every connected component of $G - X$ is exactly a clause gadget, and thus also a mock forest. Also note that $|X| = 2n$, which is polynomial in the input parameter. Finally, we let $k = n + \sum_{j = 1}^{m}(h(j) + 2)$. It remains to show that $F$  is satisfiable if and only if $(G, X, k)$ is a \yes-instance.
    
    ($\Rightarrow$) Assume the formula is satisfiable by the assignment $\phi$. We will now build an independent set $I$ in $G$ which has size at least $k$. Initially, let $I_X = \emptyset$. For each variable $x_i$, let $t_i$ be in $I_X$ if $\phi(x_i)$ is \true, and let $f_i$ be in $I_X$ otherwise. In this way, $n$ vertices are added to $I_X$. Observe that this process preserves independence of $I_X$.
    
    For each clause $C_j$, we know that there exists some satisfied literal $\ell_r$. In the corresponding clause gadget $\C_j$, observe that $s_r \notin N_G[I_X]$ by the construction of the graph and the choice of $I_X$. Then by Observation~\ref{ismf:obs:one-spike-mis}, we can choose an independent set $I_j$ for $\C_j$ which is disjoint from $N_G[I_X]$.
    
    Finally, let $I$ be the union of $I_X$ and $\bigcup_{j=1}^{m} I_j$. Observe that independence is maintained, since there are no edges between $I_X$ and $I_j$ for all $j \in [m]$, and there are no edges between $I_j$ and $I_{j'}$ for all choices of $j,j'\in [m], j \neq j'$, since there were no edges between $\C_j$ and $\C_{j'}$. Further, we note that $|I_X| = n$, and $|I_j| = h(j) + 2$ for every $j \in [m]$, and that all the sets are vertex disjoint. Thus we obtain that $|I| = n + \sum_{j=1}^{m} (h(j) + 2)$.
    
    ($\Leftarrow$) Assume that there exists an independent set $I$ for $G$ with size $|I| \geq n + \sum_{j=1}^{m} (h(j) + 2)$. We   construct an assignment $\phi$ which satisfies the SAT formula $F$. By Observation~\ref{ismf:obs:gadget-mis-size}, we know that $|I\setminus X| \leq \sum_{j=1}^{m} (h(j) + 2)$. Since $G[X]$ consists exactly of $n$ pairwise joined vertices, we also know that $|I\cap X| \leq n$. Thus, $|I| \leq n + \sum_{j=1}^{m} (h(j) + 2)$, and equality holds for all the relations. For each variable $x_i$ it must thus be the case that either $t_i \in I$ and $f_i \notin I$, or vice versa. We let $\phi(x_i)$ evaluate to \true{} if $t_i \in I$, and to \false{} otherwise.
    
    It remains to show that $\phi$ is in fact a satisfying assignment. Consider some clause $C_j$ and its corresponding gadget $\C_j$. Because $|I \cap \C_j| = h(j) + 2$, we have by Observation~\ref{ismf:obs:at-least-one-spike-mis} that there exists a spike vertex $s_r \in \C_j \cap I$. Assume for the sake of contradiction that $\ell_r \in C_j$ is not satisfied by $\phi$. This implies that $x_i$ was assigned a value that would not satisfy $\ell_r$. Without loss of generality, (by symmetry) assume $\ell_r = x_i$ and $\phi(x_i) = \false$. Then $f_i \in I$; however, by the construction of the graph, there is an edge between $s_r$ and $f_i$. This contradicts that $I$ is independent. Thus $\ell_r \in C_j$ is satisfied by $\phi$ and we have concluded the proof.
\qed\end{proof}

\begin{corollary}
    \VCMF{} does not admit a polynomial kernel unless \NPincoNPpoly{}.
\end{corollary}
 
 Finally, let us observe that in the proof of Theorem~\ref{thm:mockforest}, the graph $G-X$ is outerplanar. Also every $2$-connected component of this graph is either an edge or a cycle of length $3$. Thus the proof of Theorem~\ref{thm:mockforest} can be used to show that  \ISproblem{} parameterized by the size of a modulator to  an outerplanar graph, a cactus graph  or to  a block graph,  does does not admit a polynomial kernel.

 {\small
{ 
 \bibliographystyle{siam}
\bibliography{book_pc} 
}}



\end{document}